\documentclass[12pt]{article}

\usepackage{amsmath,amssymb,amsthm}
\usepackage{bm}
\usepackage{youngtab}
\usepackage{graphicx}

\makeatletter

\@addtoreset{equation}{section}
\makeatother

\newcommand{\beq}{\begin{equation}}
\newcommand{\eeq}{\end{equation}}

\newcommand{\bC}{\ensuremath{\mathbb{C}}}

\newcommand{\bP}{\ensuremath{\mathbb{P}}}

\newcommand{\bR}{\ensuremath{\mathbb{R}}}

\newcommand{\bT}{\ensuremath{\mathbb{T}}}

\newcommand{\bZ}{\ensuremath{\mathbb{Z}}}

\newcommand{\scC}{\ensuremath{\mathcal{C}}}

\newcommand{\scN}{\ensuremath{\mathcal{N}}}
\newcommand{\scO}{\ensuremath{\mathcal{O}}}

\newcommand{\scV}{\ensuremath{\mathcal{V}}}

\newcommand{\scZ}{\ensuremath{\mathcal{Z}}}

\newtheorem{thm}{Theorem}[section]

\newtheorem{prop}[thm]{Proposition}

\newtheorem{lem}[thm]{Lemma}

\newcommand{\mr}[1]{{\mathrm{#1}}}

\newcommand{\bb}[1]{{\mathbb{#1}}}
\newcommand{\mca}[1]{{\mathcal{#1}}}

\newcommand{\Z}{\bb{Z}}
\newcommand{\Zh}{\bb{Z}_{\mathrm{h}}}

\newcommand{\V}{\mca{V}}

\newcommand{\pg}{\overset{+}{\succ}}

\newcommand{\mg}{\overset{-}{\succ}}

\newcommand{\tenchi}{{}^{\mathrm{t}}}

\newcommand{\Vmin}{\V_{\mathrm{min}}}

\newcommand{\type}{(\sigma,\theta\sss;\sss \mu,\nu)}
\newcommand{\sss}{\hspace{0.5pt}}

\usepackage{color}
\usepackage{version}

\newenvironment{NB}{
\color{red}{\bf NB}. \footnotesize
}{}

\newenvironment{NB2}{
\color{blue}{\bf NB}. \footnotesize
}{}

\excludeversion{NB}
\excludeversion{NB2}

\begin{document}

\baselineskip=18pt  
\baselineskip 0.7cm
 
\begin{titlepage}

\setcounter{page}{0}

\renewcommand{\thefootnote}{\fnsymbol{footnote}}

\begin{flushright}
CALT-68-2755\\
IPMU09-0132\\
UT-09-24
\end{flushright}
 
\vskip 1cm

\begin{center}
{\LARGE \bf
The Non-commutative Topological Vertex \\ 
\medskip
and \\ \medskip
Wall Crossing Phenomena
}
 
\vskip 1.5cm 
 
{\large
Kentaro Nagao$^1$ and 
Masahito Yamazaki$^{2,3,4}$
}

\vskip 0.8cm

{
\it
$^1$ RIMS, Kyoto University, Kyoto 606-8502, Japan

$^{2}$  California Institute of Technology, 
CA 91125, USA\\

$^3$Department of Physics, University of Tokyo, 
Tokyo 113-0033,
Japan\\

$^4$IPMU,
University of Tokyo, 
Chiba 277-8586, 
Japan\\
}
 
\end{center}

\vspace{1cm}

\centerline{{\bf Abstract}}
\medskip
\noindent

We propose a generalization of the topological vertex, which we call  
the ``non-commutative topological vertex''. This gives open BPS invariants for a
\begin{NB}removed: an arbitrary\end{NB}%
toric Calabi-Yau manifold without compact 4-cycles, where we have D0/D2/D6-branes wrapping holomorphic 0/2/6-cycles, as well as D2-branes wrapping disks whose boundaries are on D4-branes wrapping non-compact Lagrangian 3-cycles.
The vertex is defined combinatorially using the crystal melting model proposed 
recently, and depends on the value of closed string moduli at infinity.
The vertex in one special chamber gives the same answer as that computed by the ordinary topological vertex.
We prove an identify expressing the non-commutative topological vertex of a toric Calabi-Yau manifold $X$ as a specialization of the closed BPS partition function of an orbifold of $X$, thus giving a closed expression for our vertex.
We also clarify the action of the Weyl group of an affine $A_L$
\begin{NB}corrected\end{NB}%
Lie algebra on chambers, and comment on the generalization of our results to the case of refined BPS invariants.
\end{titlepage}

\setcounter{page}{1} 


\section{Introduction and Summary}

Recently, there has been significant progress in the counting problem of BPS states in type IIA string theory on a toric Calabi-Yau 3-fold\footnote{See \cite{JM,CJ,OY1,OY2,DG,AOVY,CP}. See also \cite{Szendroi,Young1,Young2,MR,NN,Nagao1,Nagao2} for mathematical discussions}. 
In the literature, the Calabi-Yau manifold (which we denote by $X$) is assumed to have no compact 4-cycles, and we consider a BPS configuration of D0/D2-branes wrapping compact holomorphic 0/2-cycles, as well as a single D6-brane filling the entire Calabi-Yau manifold. The question is to count the degeneracy of such BPS bound states of D-branes.

One subtlety in this counting problem is the wall crossing phenomena, stating that the degeneracy of BPS bound states depends on the value of moduli at infinity. Indeed, the closed BPS partition function\footnote{The upper index $c$ stands for `closed'.}
$$
\scZ^c_{\rm BPS, (\sigma',\theta')},
$$
which is defined in \cite{OY1} as the generation function of the degeneracy of D-brane BPS bound states\footnote{The definition of the partition function $\scZ_{\rm BPS}$ is the same as the partition function $\scZ_{\rm BH}$ in \cite{OSV}.},
depends on maps $\sigma',\theta'$ specifying a chamber in the K\"ahler moduli space\footnote{See Appendix \ref{app.Weyl} for details.}.
What is interesting is that in one special chamber $\tilde{C}_{\rm top}$ of the K\"ahler moduli space, the BPS partition function is equivalent the topological string partition function\footnote{Actually, the topological string partition function depends on the choice of the resolution of the singular Calabi-Yau manifold $X$. This is related to the choice of the limit, as will be explained in the main text.} (up to the change of variables, which we do not explicitly show here for simplicity):
\beq
\scZ^c_{\rm BPS} \Big|_{\tilde{C}_{\rm top}} =\scZ^c_{\rm top} \label{ZtoZtop}.
\eeq

It is natural to expect that similar story should exist for open BPS invariants as well. Namely, we expect to define
open version of the BPS partition function\footnote{The upper index $o$ stands for `open'.}
$$
\scZ^o_{\textrm{BPS},(\sigma,\theta)}
$$
depending on maps $\sigma, \theta$ specifying the chamber in the K\"ahler moduli space, such that the partition function reduces to the open topological string partition function in a special chamber $C_{\rm top}$:
\beq
\scZ^o_{\textrm{BPS}} \Big| _{C_{\rm top}}=\scZ^o_{\rm top}.
\eeq
The question is how to define open BPS degeneracies such that the generating function follows the conditions above.

As a guiding principle of our following argument, we use the crystal melting model developed recently in \cite{OY1} (see \cite{Szendroi,MR} for mathematical discussions). 
This crystal melting model generalizes the result of \cite{ORV} for $\bC^3$ to an arbitrary toric Calabi-Yau manifold. In the case of $\bC^3$, the crystal melting partition function with the boundary conditions specified by three Young diagrams $\lambda_1,\lambda_2,\lambda_3$ gives 
the topological vertex \cite{AKMV} $C_{\lambda_1,\lambda_2,\lambda_3}$. By using these vertices
as a basic building block, we can compute open topological string partition function with non-compact D-branes wrapping Lagrangian 3-cycles of the topology $\bR^2\times S^1$ included \cite{AV}.
In this story, generalization from closed to open topological string partition function corresponds to the change of the boundary condition of the crystal melting model for $\bC^3$.

Now the recent result \cite{OY1} shows that the closed BPS partition function  discussed above can be written as a statistical mechanical partition function of the crystal model. This model applies to any toric Calabi-Yau manifold, and for $\bC^3$ the BPS partition function coincides with the topological string partition function. Similarly to the case of the topological string story mentioned in the previous paragraph, we hope to define the open version of the BPS invariants by changing the boundary condition of the crystal melting model. The invariants defined in this way will be defined in any chamber in the K\"ahler moduli space, and reduces to the ordinary topological vertex in a special chamber. We call such a generalization of the topological vertex ``the non-commutative topological vertex''\footnote{The word `non-commutative' stems from the mathematical terminologies such as ``non-commutative crepant resolution'' \cite{VdB} and ``non-commutative Donaldson-Thomas invariant'' \cite{Szendroi}. The non-commutativity here refers to that of the path algebra of the quiver. The quiver (together with a superpotential) determines a quiver quantum mechanics, which is the low-energy effective theory on the D-brane worldvolume \cite{OY1}. }, \begin{NB}added\end{NB}%
following ``the orbifold topological vertex'' named in \cite{BCY2}.\begin{NB}I removed the footnote: {See \cite{BCY2} for a related proposal. {\bf Q. Our definition include theirs?}}.\end{NB}%

We will see that this expectation is indeed true. We adopt the
\begin{NB}Removed: combinatorial\end{NB}%
definition proposed by one of the authors in the mathematical literature \cite{Nagao2,Nagao3}. Our non-commutative topological vertex is defined for a Calabi-Yau manifold $X$ without compact 4-cycles, and a set of representations $\lambda$ assigned to external legs of the toric diagram. As in the case of topological vertex, $\lambda$ encodes the boundary condition of the D4-branes wrapping Lagrangian 3-cycles. We propose our vertex as the building block of open BPS invariants. Here by an open BPS invariant we mean a degeneracy counting the number of BPS bound states of D0/D2/D6-branes wrapping holomorphic 0/2/6-cycles, as well as D2-branes wrapping disks whose boundaries are on D4-branes wrapping non-compact Lagrangian 3-cycles.

We can provide several consistency checks of our proposal (see section \ref{sec.check} for more details).
First, our vertex by definition reduces to the closed BPS invariant when all the representations $\lambda$ are trivial. Second, our vertex shows a wall crossing phenomena as we change the closed string K\"ahler moduli, and the vertex coincides with the topological vertex computation in the chamber where the closed BPS partition function reduces to the closed topological string partition function.
Third, the wall crossing factor is independent of the boundary conditions on D-branes, and is therefore the partition function factorizes into the closed string contribution and the open string contribution, as expected from \cite{OVknot} and the generalization of \cite{AOVY}.

Given a combinatorial definition of the new vertex, the next question is whether we can compute it, writing it in a closed expression. We show that the answer is affirmative, by showing the following statement.
For a Calabi-Yau manifold $X$, the non-commutative topological vertex $\scC_{\textrm{BPS},(\sigma,\theta;\lambda)}(X)$ is equivalent to the closed BPS partition function $\scZ^c_{\textrm{BPS},(\sigma',\theta')}(X')$ for an orbifold $X'$ of $X$
, under a suitable identification of variables explained in the main text\footnote{More precisely, we need to specify the resolution of $X$ and $X'$. We also need to impose the condition that two of the representations $\lambda$ are trivial. See the discussions in the main text.}:
\beq
\scC_{\textrm{BPS},(\sigma,\theta;\lambda)}(X) =\scZ^c_{\textrm{BPS},(\sigma',\theta')}(X').
\label{eq.main}
\eeq
We will give an explicit algorithm to determine $X'$ and $\sigma', \theta'$, starting from the data on the open side.
Since the infinite-product expression for $\scZ^c_{\textrm{BPS},(\sigma',\theta')}(X')$ is already known \cite{Nagao1,AOVY}, this gives a closed infinite-product expression for our vertex.

The organization of this paper is as follows. We begin in section 2 with a brief summary of the closed BPS invariants and their wall crossings, 
and their relation with the topological string theory.
In section 3 we define our new vertex using the crystal melting model. We also perform several consistency checks of our proposal. Section 4 contains our main result \eqref{eq.main}, which shows the 
equivalence of 
our new vertex with a closed BPS partition function under suitable parameter identifications. We give an explicit algorithm for 
constructing closed BPS partition function starting from our vertex.
In Section 5 we treat several examples in order to illustrate our general results. Section 6 is devoted to discussions. We also include Appendices A-C for mathematical proofs and notations. 

\section{Closed BPS Invariants}\label{sec.closed}
Before discussing the open BPS invariants, we summarize in this section the definition and the properties of the closed BPS invariants.

Throughout this paper, we concentrate on the case of 
the so-called generalized conifolds.
The reason for this is that wall crossing phenomena is understood 
well only in cases without compact 4-cycles, which means $X$ is either a generalized conifold or $\bC^3/(\bZ_2\times \bZ_2)$\footnote{See \cite{AOVY} for the proof of this statement.}. \begin{NB}Removed: The case of $\bC^3/(\bZ_2\times \bZ_2)$ is similar.\end{NB}%

By suitable $SL(2,\bZ)$ transformation, we can assume that the toric diagram of a generalized conifold is a trapezoid with height 1, with \begin{NB}corrected\end{NB}%
length $L_+$ edge at the top and $L_-$ at the bottom (see Figure \ref{gctd})\footnote{The Calabi-Yau manifold is determined by $L_+$ and $L_-$ as $xy=z^{L_+} w^{L_-}$.}. If we denote by $L=L_{+}+L_-$ the sum of the length of the edges on the top and the bottom of the trapezoid, this geometry has $L-1$ independent compact $\bP^1$'s. We label them by $\alpha_i$, borrowing the language of the root lattice of $\hat{A}_{L-1}$ algebra.

\begin{figure}[htpb]
\centering{\includegraphics[scale=0.4]{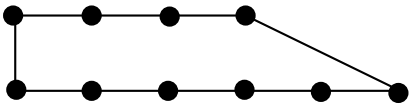}}
\caption{The toric diagram of a generalized conifold, with $L_+=3, L_-=5$.}
\label{gctd}
\end{figure}

The language of the root lattice will be used extensively throughout this paper\footnote{The root lattice of $\hat{A}_{L-1}$ is exploited in \cite{Nagao1,AOVY,Nagao2}. See also Appendix \ref{app.Weyl}.}.
We can also make more $\bP^1$'s by combining them. For example, combining all the $\bP^1$'s between $i$-th and $j$-th $\bP^1$ (assume $i<j$), we have another $\bP^1$ which we denote by 
$$
\alpha_{i,j}:=\alpha_i+\ldots+ \alpha_j. \label{alphaij}
$$
This corresponds to a positive root of $\hat{A}_{L-1}$.

Suppose that we have a Calabi-Yau manifold $X$ without compact 4-cycles.
We also consider a single D6-brane filling the entire $X$ and D0/D2-branes
wrapping compact holomorphic 0/2-cycles specified by $n \in H_0(X;\bZ)$ and $\beta\in H_2(X;\bZ)$, respectively. We can then define the
BPS degeneracy $\Omega(n,\beta)$ counting BPS degeneracy of D-branes\footnote{More precisely, this BPS degeneracy is defined by the second helicity supertrace.}.
The closed BPS partition function is then defined by
\beq
\scZ^c_{\textrm{BPS}}(q,Q)=\sum_{n,\beta}\Omega(n,\beta) q^n Q^{\beta}.
\eeq

The closed BPS partition function for generalized conifolds is studied in \cite{Nagao1,AOVY}.
To describe the results, let us first specify the resolution (crepant resolution\footnote{Crepant resolution is a resolution $f: Y\to X$ such that $\omega_Y=f^* \omega_X$, where $\omega_X$ and $\omega_Y$ are canonical bundles of $X$ and $Y$.}) of $X$\footnote{This is not essential, since by varying the value of the K\"ahler moduli we can go to the geometry with other choices of resolutions. We just need to specify an arbitrary resolution in order to begin the discussion. See Appendix \ref{app.Weyl} for more about this.}.
Each of the $L-1$ $\bP^1$'s are either $\scO(-1,-1)$-curve
or $\scO(-2,0)$-curve.
In the language of the toric diagram, this is to specify the triangulation of the toric diagram. We specify this choice by a map
\beq
\sigma: \{ 1/2,3/2,\ldots, L-1/2 \} \to \{\pm 1 \}.
\eeq
In the following we sometimes write $\pm $ instead of $\pm 1$.
When $\sigma(i-1/2)=1$ ($\sigma(i-1/2)=-1$), the $i$-th triangle from the left has 
one of its edges on the top (bottom) edge of the trapezoid.
This means that the $i$-th $\bP^1$ is a $\scO(-1,-1)$-curve ($\scO(-2,0)$-curve) when $\sigma(i-1/2) = -\sigma(i+1/2)$ ($\sigma(i-1/2) = \sigma(i+1/2)$).
By definition, we have $\big| \sigma^{-1}(\pm 1)\big|=L_{\pm}$. 

For example, in the case of Suspended Pinched Point ($L_+=1, L_-=2$) whose toric diagram is shown in Figure \ref{fig.sigmaex}, $L=3$ and there are 3 difference choice of resolutions. This is represented by
\begin{align}
\sigma_1:\{ 1/2,3/2,5/2\}\to \{-,-,+ \}, \\
\sigma_2:\{ 1/2,3/2,5/2\}\to \{-,+,- \}, \\ 
\sigma_3:\{ 1/2,3/2,5/2\}\to \{+,-,- \}.
\label{SPPsigma}
\end{align}

\begin{figure}[htpb]
\centering{\includegraphics[scale=0.4]{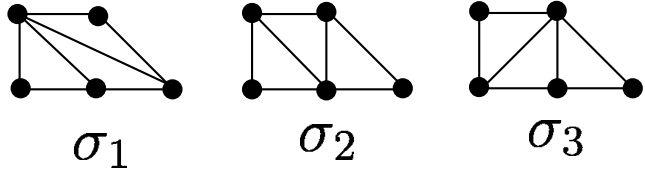}}
\caption{The choice of resolutions of a generalized conifold ($L_+=1,L_-=2$).}\label{fig.sigmaex}
\end{figure}

Given $\sigma$, the topological string partition function is given by \cite{AKMV,IK}
\beq
\scZ_{\textrm{top}, \sigma}(q=e^{-g_s},Q=e^{-t})=\prod_{n=1}^{\infty} (1-q^n Q)^{n N_{\beta}^{g=0}},
\eeq
where $N_{\beta}^{g=0}$ is the genus 0 Gopakumar-Vafa (GV) invariant\footnote{Higher genus GV invariants vanish for generalized conifolds.}. For the 2-cycle $\beta=\alpha_i+\ldots +\alpha_j$, the explicit form of $N_{\beta}^0$ depends on $\sigma$ and is given by
\begin{align*}
N_{\beta=\alpha_i+\ldots \alpha_j}^{g=0}&=(-1)^{1+\sharp \{k\mid i\le k\le j, \ \sigma(k-1/2)\ne \sigma(k+1/2) \}} \\
&=(-1)^{1+\sharp \{k\mid i\le k\le j, \ \alpha_k \textrm{ is a } \scO(-1,-1)-\textrm{curve} \}}.
\end{align*}
By CPT invariance in five dimensions \cite{AOVY}, we have
$
N_{\beta=-(\alpha_i+\ldots \alpha_j)}^{g=0}
=N_{\beta=\alpha_i+\ldots \alpha_j}^{g=0}.
$
We also have $N_{\beta=0}^0=\chi(X)/2$, where the Euler character $\chi(X)$ for a toric Calabi-Yau manifold is the same as twice the area of the toric diagram.

As shown in \cite{AOVY,Nagao1}, the closed BPS partition function is given by
%
\begin{eqnarray}
\scZ_{\rm BPS}(q,Q)&=&\scZ_{\rm top}(q,Q) \scZ_{\rm top}(q,Q^{-1}) \big|_{\rm chamber} \nonumber \\
&=& \prod_{(\beta,n): Z(\beta,n)>0} (1-q^n Q^{\beta})^{n N_{\beta}^0}, \label{Zchamber}
\end{eqnarray}
where the central charge $Z(\beta,n)$ is given by
$$
Z(\beta,n)=(B(\beta)+n)/R. \label{Z}
$$
Here $1/R$ denotes (up to proportionality constants) the central charge of the D0 brane, and following \cite{AOVY} we choose the complexified K\"ahler moduli to be real. Also, the notation $B(\beta)$ means the B-field flux through the cycle $\beta$\footnote{This was written $\beta B$ in \cite{AOVY}.}.

Now suppose that $1/R$ is positive\footnote{Under this condition we are discussing only half of chambers of the K\"ahler moduli space, which lie between the Donaldson-Thomas chamber and the non-commutative Donaldson-Thomas chamber.
 The other half arises when $1/R$ is negative.}.
From \eqref{Zchamber} and \eqref{Z}, it follows that the wall crossing occurs when the integer part of the value of the B-field through the cycle change. For the cycle $\alpha_{i}+\ldots +\alpha_j$, this is given by
$$
\left[ B(\alpha_i)+\ldots + B(\alpha_j) \right], \label{Balphaij}
$$
Since there are \begin{NB}corrected\end{NB}%
$L-1$ $\bP^1$'s in $X$, there are \begin{NB}corrected\end{NB}%
$L(L-1)/2$ such parameters.

We can take a special limit $B(\alpha_i)\to \infty$. Let us denote this special chamber by $\tilde{C}_{\rm top}$. As discussed in \cite{AOVY,Nagao1}, in this limit the BPS partition function reduces to the closed topological string partition function:
$$
\scZ^c_{(\sigma,\theta)}\Big|_{\tilde{C}_{\rm top}} =\scZ^c_{\rm top},
$$
just as advertised in \eqref{ZtoZtop}.

\medskip
For concreteness, let us discuss an example.
We use the example of the Suspended Pinched Point ($N=3$) using the triangulation $\sigma_1$ in \eqref{SPPsigma}. In this example, the topological string partition function is 
$$
\scZ_{\textrm{top}, \sigma=\sigma_1}(q,Q)=M(q)^{3/2} \prod_{n=1}^{\infty} (1-q^n Q_1)^{-n}
\prod_{n=1}^{\infty} (1-q^n Q_2)^{n} \prod_{n=1}^{\infty} (1-q^n Q_1 Q_2)^{n},
$$
where $M(q)$ is the MacMahon function
$$
M(q)=\prod_{n=1}^{\infty} (1-q^n)^{-n}.
$$
The BPS partition function is given by
\begin{align*}
\scZ_{\rm BPS}&(q,Q)=M(q)^{3} 
\prod_{n=1}^{\infty} (1-q^n Q_1)^{-n}
\prod_{n=1}^{\infty} (1-q^n Q_2)^{n} \prod_{n=1}^{\infty} (1-q^n Q_1 Q_2)^{n} \\%
&\times 
\prod_{n>\left[ B(\alpha_1) \right]}^{\infty} (1-q^n Q_1^{-1})^{-n}
\prod_{n>\left[ B(\alpha_2) \right]}^{\infty} (1-q^n Q_2^{-1})^{n} \prod_{n>\left[ B(\alpha_1+\alpha_2) \right]}^{\infty} (1-q^n (Q_1 Q_2)^{-1})^{n}.
\end{align*}


\medskip

The parameters $\left[ B(\alpha_{i}+\ldots +\alpha_j) \right]$ specify the chamber, but as we can see from the definition they are not completely independent parameters. Since we only have \begin{NB}corrected\end{NB}%
$L-1$ real parameters $B_i$, it is likely that this parametrization is redundant. Indeed, as explained in the Appendix \ref{app.Weyl} 
we can specify the chamber by a map $\theta$, which is specified by \begin{NB}corrected\end{NB}%
$L$ half-integers, \begin{NB}corrected\end{NB}%
$\theta(1/2),\theta(3/2),\ldots, \theta(L-1/2)$, satisfying one constraint
$$
\sum_{i=1}^L \theta\left(i-\frac{1}{2}\right)=\sum_{i=1}^L \left(i-\frac{1}{2}\right).
$$
This means we can indeed parametrize the chamber by \begin{NB}corrected\end{NB}%
$L-1$ independent (half-)integers, which is what we expected. As discussed in Appendix \ref{app.Weyl}, $\theta$ is an element of the Weyl group of \begin{NB}corrected\end{NB}%
$\hat{A}_{L-1}$.

\section{The Noncommutative Topological Vertex}\label{sec.def}

In this section we give a general definition of the non-commutative topological vertex using the crystal melting model. This definition is equivalent to the one given in \cite{Nagao2} using the dimer model\footnote{See Appendix A of \cite{OY1} for the equivalence between crystal melting model and the dimer model.}. \begin{NB}added\end{NB}%
See \cite{Nagao3} for more conceptual definition in terms of Bridgeland's stability conditions and moduli spaces. 


To define our vertex, we need the following set of data:

\begin{itemize}
\item A map 
$$
\sigma: \{ 1/2,\ldots, L-1/2\}\to \{\pm \}.
$$ 
As already explained in section \ref{sec.closed}, this gives a triangulation of the toric diagram, or equivalently the choice of the resolution of the Calabi-Yau manifold. 

\item A map $\theta:\bZ_h \to \bZ_h$. As explained in Appendix \ref{app.Weyl} in the case of closed BPS invariants, $\theta$ and $\sigma$ specify the chamber structure of the open BPS invariants. 

\item A set of Young diagrams $\lambda$, assigned to external legs of the $(p,q)$-web. This specifies the boundary condition of the non-compact D-branes ending on the $(p,q)$-web. We denote by \begin{NB}corrected\end{NB}%
$\lambda_1,\ldots,\lambda_{L}$ the Young diagrams for the top and the bottom edges of the trapezoid, and by $\lambda_+,\lambda_-$ the remaining two. We sometimes write $\lambda=(\mu,\nu)$, where \begin{NB}corrected\end{NB}%
$\mu=(\mu_1,\ldots, \mu_L)=(\lambda_1,\ldots,\lambda_L)$ and $\nu=(\lambda_+,\lambda_-)$.
In the example shown in Figure \ref{lambdaconv}, there are five external legs and we have five representations.

\begin{figure}[htbp]
\centering{\includegraphics[scale=0.4]{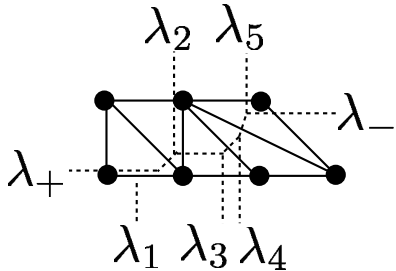}}
\caption{representations assigned to external legs of the $(p,q)$-web. The dotted lines represent the $(p,q)$-web.}
\label{lambdaconv}
\end{figure}

For later purposes, we combine $\mu_1,\ldots, \mu_L$ into a single representation $\mu$ by
\beq
\mu(i-1/2+kL)=\mu_i(k-1/2). \label{combinedmu} 
\eeq
In other words, we choose $\mu$ such that L-quotients of $\mu$ give $\mu_1, \ldots, \mu_L$. By abuse of notation, we use the same symbol $\mu$ for a set of representations $\mu_1,\ldots, \mu_L$ as well as a single representation define above.
\end{itemize}

Given $\sigma, \theta$ and $\lambda$, we define the non-commutative topological vertex
$$
\scC_{\textrm{BPS},(\sigma,\theta;\lambda)}(q,Q).
$$
In the following we drop the subscript BPS for simplicity.

Before going into the general definition, we first illustrate our idea using simple example of the resolved conifold.


\subsection{Example: Resolved Conifold}\label{subsec.conifoldeg}

In this example there is only one $\bP^1$ and the BPS partition function depends on a single positive integer $N:=\left[ B(\alpha_1) \right]$. In the language of $\theta$,\begin{NB}I change the notation: $\theta\to\theta^{-1}$\end{NB}%
$$
\theta(1/2)=1/2-N, \quad \theta(3/2)=3/2+N. \label{thetaN}
$$
We fix $\sigma$ to be 
$$
\sigma(1/2)=+, \quad \sigma(3/2)=-.
$$ 
Without losing generality we concentrate on $N\ge 0$, since $N<0$ corresponds to a flopped geometry, where $\sigma$ is replaced by $-\sigma$ (see Appendix \ref{app.Weyl}).

The ground state crystal for $N=2$ is shown in Figure \ref{conifoldcrystal}.
This crystal, sometimes called a pyramid, consists of infinite layers of atoms, the color alternating between black and white \begin{NB}added\end{NB}%
(\cite{Szendroi, Young1}). 
In the $N$-th chamber there are $N+1$ atoms on the top.

\begin{figure}[htbp]
\centering{\includegraphics[scale=0.25]{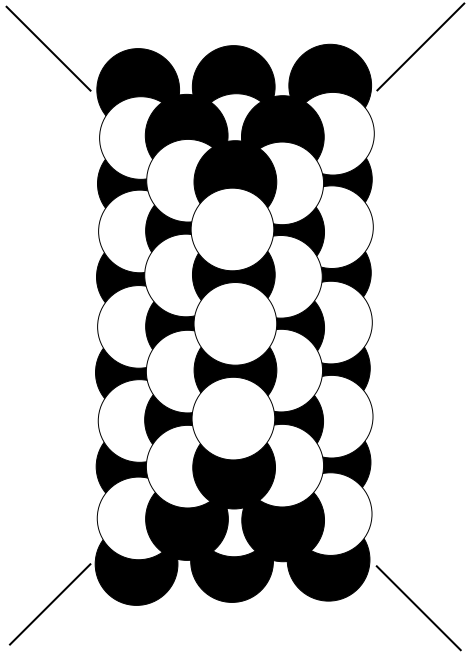}}
\caption{The ground state crystal for the resolved conifold for $N=2$. The crystal consists of an infinite number of layers, and only a finite number is shown here. The ridges of the pyramid are represented by four lines extending to infinity.}
\label{conifoldcrystal}
\end{figure}

The closed BPS partition function is defined by removing a finite set of atoms $\Omega$ from the crystal. When we do this, we follow the melting rule \cite{OY1,MR} such that
whenever an atom is removed from the crystal, we remove all the atom above it.
In other words, since an atom is in one-to-one correspondence with an F-term equivalent class of paths starting from a fixed node of the quiver diagram \cite{OY1}, 
for an arrow $a$ and and an atom $\alpha$ we can define $a\alpha$. The melting rule then says 
\beq
{\rm If } \quad a\alpha\in \Omega,\quad {\rm then} \quad \alpha\in \Omega.
\label{meltingrule}
\eeq

We then define the partition function by summing over such $\Omega$:
\beq
\scZ=\sum_{\Omega} (q_0^{(N)})^{w_0(\Omega)} (q_1^{(N)})^{w_1(\Omega)},
\label{Zdefconifold}
\eeq
where $w_0(\Omega)$ and $w_1(\Omega)$ are the number of white and black atoms in $\Omega$, respectively.

The weights $q_0^{(N)}$ ($q_1^{(N)}$) assigned to white (black) atoms in the $N$-th chamber are determined as follows. We can slice the crystal by the plane, 
and each slice is specified by an integer $i$ (see Figure \ref{crystalslice}). We choose $i$ so that atoms on the top of the crystal is located at $i=0$.

\begin{figure}[htbp]
\centering{\includegraphics[scale=0.25]{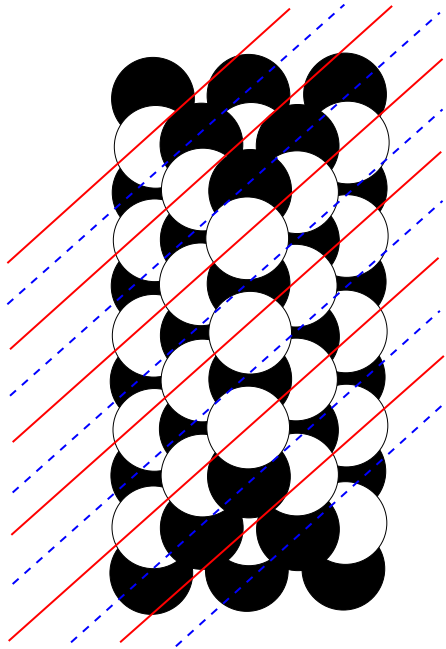}}
\caption{We can slice of the conifold crystal by an infinite number of parallel planes.}
\label{crystalslice}
\end{figure}

The weight $q_i^{(N)}$ depends on the chamber and is given by\begin{NB}corrected\end{NB}%
\beq
 q_0^{(N)}=q_0^{-N}q_1^{-N+1}, 
\quad
 q_1^{(N)}=q_0^{N+1}q_1^{N},
\label{conifoldwt1}
\eeq
when $N$ is odd, and
\beq
 q_0^{(N)}=q_0^{N+1}q_1^{N},
\quad
 q_1^{(N)}=q_0^{-N}q_1^{-N+1},
\label{conifoldwt2}
\eeq
when $N$ is even.
For example, $q_0^{(0)}=q_0, q_1^{(0)}=q_1$ when $N=0$, and \begin{NB}corrected\end{NB}%
$q_0^{(1)}=q_0^{-1} , q_1^{(1)}=q_0^2q_1$.
The change of variables arises from the Seiberg duality on the quiver quantum mechanics \cite{CJ}, geometrically mutations in the derived category of coherent sheaves \cite{CJ,NN}, or in more combinatorial language the dimer shuffling \cite{Young1}. The parameters $q_0, q_1$ defined here are related to the D0/D2 chemical potentials introduced in section \ref{sec.closed} by\footnote{
The equation \eqref{conifoldwt2} is the same for $N$ odd and even if we suitably exchange the two nodes of the quiver diagram. 
The relation \eqref{chemicalpot} can also be written as
$$
q=q_0^{(N)} q_1^{(N)}, \quad Q=(q_0^{(N)})^{N} (q_1^{(N)})^{N+1}, 
$$
when $N$ is even, and $q_0^{(N)}$ and $q_1^{(N)}$ exchanged when $N$ odd. This coincides with the expression in \cite{CJ}.
}
\beq
q=q_0 q_1, \quad Q_1=q_1.
\label{chemicalpot}
\eeq

Now let us discuss the open case. When non-trivial representations are assigned to each of the four external legs of the $(p,q)$-web, the only thing we need to do is to change the ground state of the crystal.

The crystal has four ridges, corresponding to four external legs of the $(p,q)$-web. When we assign a representation, we remove the atoms with the shape of the Young diagram. More precisely, we remove the atoms with the shape of the Young diagram in the asymptotic direction of the $(p,q)$-web, as well as all the atoms above them, so that the melting rule is satisfied. See Figure \ref{openconifoldeg} for an example.

\begin{figure}[htbp]
\centering{\includegraphics[scale=0.25]{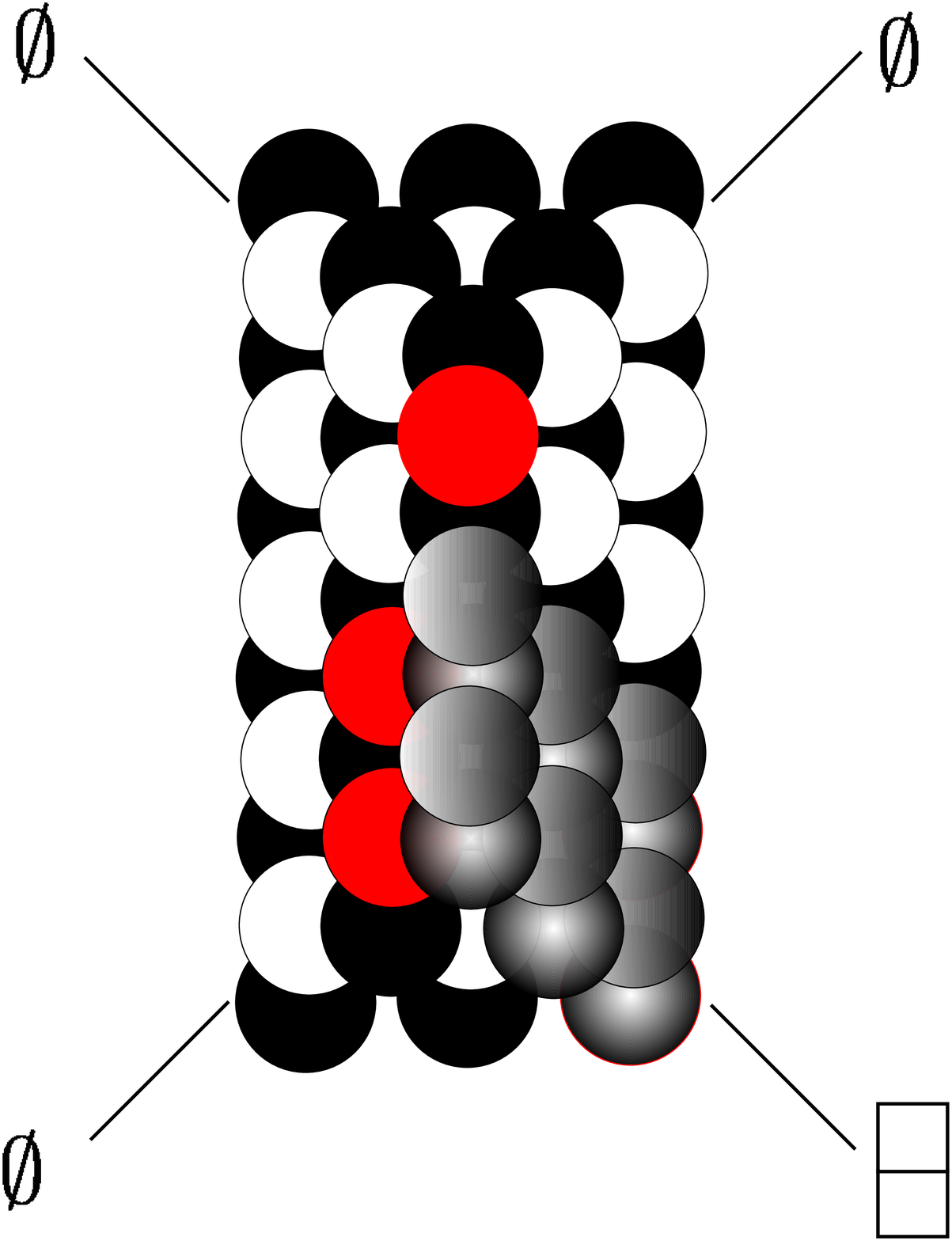}}
\caption{The pyramid for open BPS invariants. A non-trivial representation 
$(1,1)$ is placed on with one of the four external lines. As compared with the previous figure, atoms colored gray, corresponding to the Young diagram, are removed from the crystal. The red atoms have no atoms above them.}
\label{openconifoldeg}
\end{figure}

The partition function is defined in exactly the same way by \eqref{Zdefconifold}, and the result is denoted by $\scC_{(\sigma, \theta;\lambda)}$.

\bigskip

Several comments are now in order.

\begin{NB}Added\end{NB}%
First, let us explain the origin of the name ``the non-commutative topological {\it vertex}''.
Recall that, in commutative case, topological vertex is defined for $\bC^3$. 
For a general affine toric Calabi-Yau manifold $X$, we divide the polygon into triangles and assign a topological vertex to each trivalent vertex of the dual graph. 
We can get the topological string partition function for the smooth toric Calabi-Yau manifold $Y$ by gluing them with propagators. 
Similarly, \begin{NB2} added \end{NB2} assume that a polygon is divided into trapezoids. 
Then we can assign a non-commutative topological vertex to each vertex of the dual graph and glue them by propagators. 
The BPS partition function defined in this way \begin{NB2} changed \end{NB2} is related to the 
topological string partition function 
via wall-crossing\footnote{
Given a devision of a polygon into trapezoids, we get a partial resolution of $X$ and a non-commutative algebra $A$ over the partial resolution, which is derived equivalent to $Y$. The BPS partition function 
given by gluing non-commutative topological vertices counts torus invariants $A$-modules.}.
In \cite{BCY2,BCY1} they study the case when a polygon is divided into (not necessary minimal) triangles.

\begin{NB}Removed: First, let us emphasize the difference between ordinary (`commutative') topological vertex and our non-commutative topological vertex. 
In commutative case, topological vertex is defined only for $\bC^3$. For a more general toric Calabi-Yau manifold, we cut the $(p,q)$-web into pieces,
assign a topological vertex to each trivalent vertex and 
we glue them by propagators. For example, for the resolved conifold we need to glue two vertices. 
However, our NCTV is defined for any toric Calabi-Yau manifold, and thus we do not need to glue them. In this sense, our vertex is not a `vertex'.
\end{NB}

Second, it is possible to give more geometric definition of the vertex \begin{NB}Added\end{NB}%
(see \cite{Nagao3}).
For the closed BPS invariants, 
the crystal arises as a torus fixed point of the moduli space of the modules of the path algebra quiver (under suitable $\theta$-stability conditions). The moduli space is the vacuum moduli space of the quiver quantum mechanics arising as the low-energy effective theory of D-branes \cite{OY1}. The similar story exists in our case. Namely, the crystal is in one to one correspondence with the fixed point of the moduli space arising from a quiver diagram. For example, for conifold with $\lambda=${\tiny \yng(1,1)}, the quiver is given in Figure \ref{openquiver}.

\begin{figure}[htbp]
\centering{\includegraphics[scale=0.5]{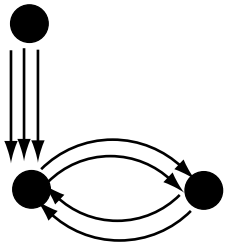}}
\caption{Quiver diagram for the open invariant with 
$\lambda=(1,1)$. This is the Klebanov-Witten quiver \cite{KW} with an extra node and extra three arrows starting from it. The three arrows correspond to three red atoms in Figure \ref{openconifoldeg}.}
\label{openquiver}
\end{figure}

Third, in the case of $\bC^3$, our vertex reproduces the topological vertex of $\bC^3$
by definition.


\subsection{General Definition from Crystal Melting}

We next give a general definition of the vertex. Readers not interested in the details of the definition of the non-commutative topological vertex
can skip this section on first reading.

Given a boundary condition specified by $\sigma, \theta$ and $\lambda=(\mu,\nu)$, we would like to construct a ground state of the crystal, and
determine the weights assigned to the atoms of the crystal.

The basic idea is the same as in the conifold example. First, the closed string BPS partition function is equivalent to the statistical partition function of crystal melting. The ground state crystal can be sliced by an infinite number of parallel planes parametrized by integers $n\in \bZ$, just as in Figure \ref{crystalslice}. On each slice, there are infinitely many atoms, labeled by integers $(x,y)\in \bZ_{\ge 0}^2$. Therefore, the atoms in the crystal are label by $(n,x,y)\in \bZ\times \bZ_{\ge 0}^2$.

\bigskip

Let us show this in the example of the Suspended Pinched Point. The crystal in
Figure \ref{SPPcrystal} clearly shows this structure.

\begin{figure}[htbp]
\centering{\includegraphics[scale=0.2]{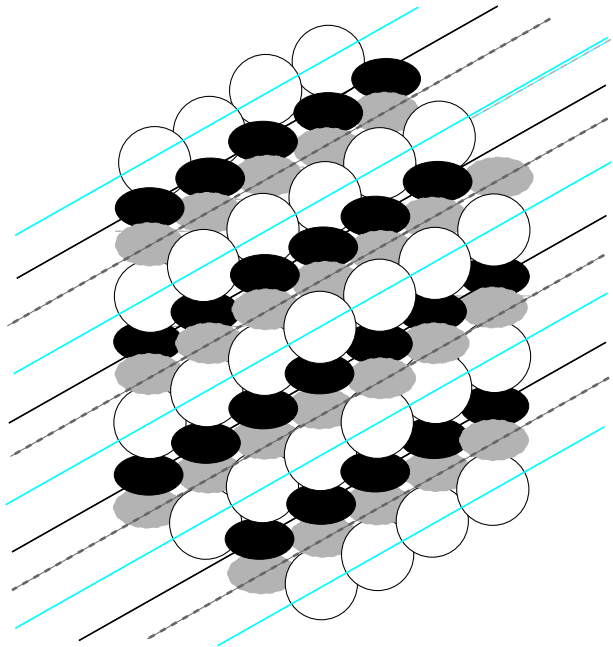}}
\caption{The crystal for the Suspended Pinched Point. We can slice the crystal along planes represented by lines, which come with three different colors.}
\label{SPPcrystal}
\end{figure}

Another way of explaining this is to 
construct a crystal starting from a bipartite graph on $\bR^2$, shown in Figure \ref{SPPbipartite}\footnote{This is a universal cover of the bipartite graph on $\bT^2$, which appears in the study of four-dimensional $\scN=1$ quiver gauge theories. See \cite{BT1,BT2,BT3} for original references, and \cite{Kennaway,Yamazaki} for reviews.}.
In this example, the bipartite graph consists of hexagons and squares, and 
periodically changes its shape along the horizontal directions.

\begin{figure}[htbp]
\centering{\includegraphics[scale=0.25]{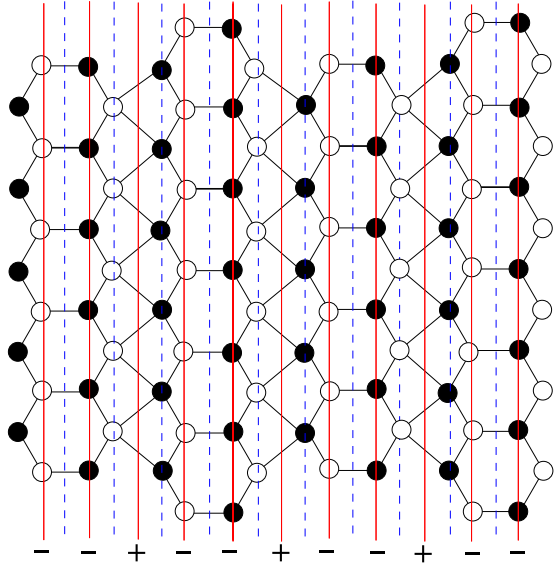}}
\caption{The bipartite graph for Suspended Pinched Point. We here take $\sigma=\sigma_1$ and $\theta=id$. The red undotted (the blue dotted) lines have half-integer (integer) values of the coordinate along the horizontal axis.}
\label{SPPbipartite}
\end{figure}

Now the atoms of the crystal are located at the centers of the faces of the bipartite graph. and it thus follows we can slice the crystal along the horizontal axis. Each slice consist of an infinite number of atoms labeled by two integers $(x,y)\in (\bZ_{\ge 0})^2$, since there are two directions, the horizontal direction and the perpendicular direction to the paper\footnote{In general the bipartite graph is determined by $\sigma$. 
A hexagon (a square) corresponds to $\scO(-2,0)$-curve ($\scO(-1,-1)$-curve).
In other words, the $i$ th polygon is a hexagon (square) if 
$\sigma(i)=\sigma(i+1) (\sigma(i)\ne \sigma(i+1))$.}.

\bigskip

Now consider the open case. In this case, we construct a new ground state by removing atoms from the closed ground state. By the melting rule, the atoms 
removed from the $n$-th plane should be labeled by $(x,y)\in \scV(n)$, where
$\scV(n)$ is a Young diagram. Depending on the representations on external legs, $\scV(n)$ increases or decreases as we change $n$. Thus the ground state crystal for open BPS invariants are determined by such a sequence of Young diagrams $\{\scV(n)\}$, called transitions below. In the following we make this idea more rigorous.

Let us begin with some notations.
Let $\mu$ and $\mu'$ be two Young diagrams. 
We say $\mu\pg\mu'$ if the row lengths satisfy
\[
\mu_1\geq\mu'_1\geq\mu_2\geq\mu'_2\geq\cdots,
\]
and $\mu\mg\mu'$ if the column lengths satisfy
\[
\tenchi\mu_1\geq\tenchi\mu'_1\geq\tenchi\mu_2\geq\tenchi\mu'_2\geq\cdots.
\]
We define a {\it transition} $\V$ of Young diagrams of type $\type$ as a map 
from the set of integers $\bZ$ to 
the set of Young diagrams such that
\begin{itemize}
\item $\V(n)=\nu_-$ for $n\ll0$ and $\V(n)=\nu_+$ for $n\gg0$,
\item $\V(h-\mu\circ\theta(h)/2)\overset{\sigma\circ\theta(h)}{\succ}\V(h+\mu\circ\theta(h)/2)$.
\end{itemize}
Then as shown in \cite{Nagao2} there is a transition $\Vmin$ of Young diagrams of type $\type$ such that for any transition $\V$ of Young diagrams of type $\type$  and for any $n\in\Z$ we have $\V(n)\supseteq \Vmin(n)$. 
The transition $\Vmin(n)$ is the sequence of transitions discussed above.

For a transition $\V$ of Young diagram of type $\type$, the ground state crystal can be defined by
\[
A(\V):=\{a(n,x,y)\mid n\in \bZ, x,y \in \bZ_{\ge 0}, \ (x,y)\notin \V(n)\}.
\]
where $a(n,x,y)$ denotes the atom at position $(n,x,y)$.

Having defined the ground state crystal, the partition functions is defined again as the sum over a subset $\Omega$ of $A(\Vmin)$ satisfying the following two conditions:\footnote{It is straightforward to show that a subset $\Omega \subset A(\Vmin)$ satisfies the two conditions if and only if $P(\Vmin)\backslash \Omega=A(\V)$ for a transition $\V$ of Young diagram of type $\type$.}

\begin{itemize}
\item $\Omega$ is finite set, and 
\item $\Omega$ satisfies the melting rule \eqref{meltingrule}. In other words, if 
$a'\in \Omega$ and $a'=a \alpha$ for an arrow $\alpha$, then $a\in \Omega$.\footnote{
We can also explicitly write down the melting rule using the coordinates $(n,x,y)$.
Let us write $a \sqsupset a'$ when there is a path (a composition of arrows) $\alpha$ such that $a'=a \alpha$. 
The partial order $\sqsupset$ is then generated by
$$
\begin{array}{lr}
a(h-1/2,x,y)\sqsupset a(h+1/2,x,y) & \lambda\circ \theta(h)=+,\\
a(h+1/2,x,y)\sqsupset a(h-1/2,x,y) & \lambda\circ \theta(h)=-,\\
a(h+1/2,x,y)\sqsupset a(h-1/2,x+1,y) & \lambda\circ \theta(h)=+, \sigma\circ\theta(h)=-,\\
a(h+1/2,x,y)\sqsupset a(h-1/2,x,y+1) & \lambda\circ \theta(h)=+, \sigma\circ\theta(h)=+,\\
a(h-1/2,x,y)\sqsupset a(h+1/2,x+1,y) & \lambda\circ \theta(h)=-, \sigma\circ\theta(h)=-,\\
a(h-1/2,x,y)\sqsupset a(h+1/2,x,y+1) & \lambda\circ \theta(h)=-, \sigma\circ\theta(h)=+,\\
a(n,x,y)\sqsupset a(n,x+1,y) & \sigma\circ\theta(n-1/2)=\sigma\circ\theta(n+1/2)=+,\\
a(n,x,y)\sqsupset a(n,x,y+1) & \sigma\circ\theta(n-1/2)=\sigma\circ\theta(n+1/2)=-.
\end{array}
$$
}
\end{itemize}

For a crystal $\Omega \in A(\Vmin)$, we define the weight $w(\Omega)_i$ by the number of atoms with the color $i$ contained in $\Omega$:
\[
w(\Omega)_i:=\sharp\{a(n,x,y)\in \Omega\mid n\equiv i\ (\mr{mod}\,L)\}.
\]
Also, for $\theta$, we put 
\[
q^\theta_i:=
\begin{cases}
q_{\theta^{-1}(i-1/2)+1/2}\cdot q_{\theta^{-1}(i-1/2)+3/2}\cdot\cdots\cdot q_{\theta^{-1}(i+1/2)-1/2} &  (\theta^{-1}(i-1/2)<\theta^{-1}(i+1/2)),\\
q_{\theta^{-1}(i-1/2)-1/2}^{-1}\cdot q_{\theta^{-1}(i-1/2)-3/2}^{-1}\cdot\cdots\cdot q_{\theta^{-1}(i+1/2)+1/2}^{-1} & (\theta^{-1}(i-1/2)>\theta^{-1}(i+1/2)),
\end{cases}
\]
where we define $q_i$ for $i\in \bZ$ periodically,
$$
q_{i+L}=q_i.
$$
We then define the vertex by
\[
\mathcal{C}^{\rm ref}_{\type}(q_0,\ldots,q_{L-1}):=\sum_{\Omega} (q^{\theta}_0)^{w(\Omega)_0}\cdot\cdots\cdot (q^{\theta}_{L-1})^{w(\Omega)_{L-1}}.
\]
The parameters $q_0, \ldots, q_{L-1}$ defined here are related to the D0/D2 chemical potentials introduced in section \ref{sec.closed} by
\beq
q=q_0\ldots q_{L-1}, \quad Q_i=q_i\quad (i=1,\ldots, L-1).
\label{chemicalweight}
\eeq

\begin{NB2} added the whole section \end{NB2}

\subsection{Refinement}\label{subsec.refined}

We can also generalize the definition to include the open refined BPS invariants\footnote{See \cite{KontsevichS,DG,BBS} for closed refined BPS invariants}. 

Let us recall the meaning of the refined BPS counting, first in the closed case.
When the type IIA brane configuration is lifted to M-theory \cite{GV1} and when we use the 4d/5d correspondence \cite{GSY,DVVafa}, the D0/D2-branes are lifted to spinning M2-branes in $\bR^5$, which has spin under the little group in 5d, namely $SO(4)=SU(2)_L\times SU(2)_R$. The ordinary BPS invariant is defined as an index; it keeps only the $SU(2)_L$ spin, while taking an alternate sum over the $SU(2)_R$ spin. 
The refined closed BPS invariants is defined by taking both spins into account.

The situation changes slightly when we consider open refined BPS invariants.
The D4-branes wrapping Lagrangians, when included, are mapped to M5-branes on $\bR^3$. This means that $SO(4)$ is broken to $SO(2)$, and we have only one spin. However, there is an $SO(2)$ R-symmetry for $\scN=2$ supersymmetry in three dimensions, and in the definition of the ordinary open BPS invariants we keep only one linear combination of the two, while tracing out the other combination \cite{LMV}. The refined open BPS invariants studied here takes boths of the two charges into account.

In the language of crystal melting used in this paper, the open refined BPS invariants are defined simply by modifying the definition of the weights. Here, we explain how to modify the weights in the case of $\mu=\emptyset$.

For an integer $n$ with 
$$n\equiv \pi(n) \quad \textrm{mod}\  L, \quad 0\le \pi(n)\le L-1,$$ define
\[
w(\Omega)_n:=\sharp\{ a(n,x,y)\in \Omega \}.
\]
We also define the weights by
\[
\tilde{q}^\theta_n:=
\begin{cases}
\tilde{q}_{\theta^{-1}(n-1/2)+1/2}\cdot \tilde{q}_{\theta^{-1}(n-1/2)+3/2}\cdot\cdots\cdot \tilde{q}_{\theta^{-1}(n+1/2)-1/2} &  (\theta^{-1}(n-1/2)<\theta^{-1}(n+1/2)),\\
\tilde{q}_{\theta^{-1}(n-1/2)-1/2}^{-1}\cdot \tilde{q}_{\theta^{-1}(n-1/2)-3/2}^{-1}\cdot\cdots\cdot \tilde{q}_{\theta^{-1}(n+1/2)+1/2}^{-1} & (\theta^{-1}(n-1/2)>\theta^{-1}(n+1/2)),
\end{cases}
\]
where
$$
\tilde{q}_n=q_{\pi(n)}
$$
when $n \not\equiv 0 \ \textrm{mod} \ L$,
and
$$
\tilde{q}_n=
\begin{cases}
q_+ & n>0 \\
(q_+ q_-)^{1/2}&  n=0\\
q_- & n<0
\end{cases}
$$
when $n \equiv 0 \ \textrm{mod} \ L$.
We then define the refined vertex by
\[
\mathcal{C}_{(\sigma,\theta\sss;\sss \emptyset,\nu)}(q_+,q_-,q_1\ldots,q_{L-1}):=\sum_{\Omega} \prod_{n\in \bZ} (\tilde{q}^{\theta}_n)^{w(\Omega)_n}.
\]
By definition, the refined vertex reduces to the unrefined vertex by setting $q_+=q_-=q_0$.
The reader can refer to \cite{Nagao2} for the definition of weights in general cases $\mu\ne \emptyset$.

\subsection{Consistency Checks of Our Proposal}\label{sec.check}

In Appendix \ref{app.proof} we gave a purely combinatorial definition of the non-commutative topological vertex. We now claim that is captures open BPS invariants in the following sense.

Consider a generalized conifold (a toric Calabi-Yau manifold without compact 4-cycles) with representations assigned to each leg of the $(p,q)$-web. Each representation specifies a boundary condition on the non-compact D4-brane wrapping Lagrangian 3-cycle of topology $\bR^2\times S^1$ \cite{AV}.

In the absence of D4-branes, we are counting particles of D0/D2-branes wrapping 0/2-cycles, which makes a bound state with a single D6-brane filling the entire Calabi-Yau manifold. When the D4-branes are included, D2-branes can wrap disks ending on the worldvolume of D4-branes. The degeneracy of such D-brane configurations is what we mean by the open BPS degeneracies. Note that supersymmetry is broken by half due to the inclusion of D4-branes; our counting of BPS particles makes sense because we are counting BPS states in lower dimensions, where the minimal amount of supersymmetry is lower.

We can provide several consistency checks of our proposal.
First, our vertex by definition reduces to the closed BPS invariant when all the representations $\lambda$ are trivial:
$$
\scC_{\textrm{BPS}, (\sigma,\theta;\lambda=\emptyset)}=\scZ^c_{\textrm{BPS},(\sigma,\theta)}.
$$
The second consistency check comes from the wall crossing phenomena. As shown in \cite{Nagao2}, the vertex goes through a series of wall crossings as we move around the closed string moduli space (K\"ahler moduli of the Calabi-Yau manifold), just as in the case of closed invariants. 
It was also shown in \cite{Nagao2} that in the chamber $C_{\rm top}$ where the closed BPS partition function reduces to the closed topological string partition function, our vertex gives the same answer as that computed from the topological vertex (in the standard framing):
$$
\scC_{\textrm{BPS}, (\sigma,\theta;\lambda)} \Big|_{C_{\rm top}}=\scC_{\textrm{topological vertex}, \lambda}.
$$
The third consistency check comes from the fact that the wall crossing factor is independent of representations. In other words,
\[
\bar{\mathcal{C}}_{\type}(q_0,\ldots,q_{L-1}):=
\frac{\mathcal{C}_{\type}(q_0^\theta,\ldots,q_{L-1}^\theta)}{\mathcal{C}_{(\sigma,\theta\sss;\sss\emptyset,\emptyset)}(q_0^\theta,\ldots,q_{L-1}^\theta)}.
\]
does not depend on $\theta$ \cite{Nagao2,Nagao3}\footnote{\label{normalization}To be exact, we have to normalize the generating function by a monomial. See \cite[Corollary 3.21]{Nagao3} for the precise statement.}. This means that the open BPS partition function, which is defined by the sum over representations, takes a factorized form 
\beq
\scZ^{o}_{\rm BPS} =\frac{\scZ^{o}_{\rm BPS}}{\scZ^{o}_{\rm top}} \scZ^o_{\rm top}
= \frac{\scZ^{c}_{\rm BPS}}{\scZ^{c}_{\rm top}} \scZ^o_{\rm top}.
\label{Zoproduct}
\eeq
Since $\scZ^{c}_{\rm BPS}/\scZ^{c}_{\rm top}$ takes an infinite product form as explained in section \ref{sec.closed} and $\scZ^o_{\rm top}$ also takes the infinite product form \cite{OVknot}, $\scZ^{o+c}_{\rm BPS}$ itself should take an infinite product form,
which is consistent with a suitable generalization of the argument of \cite{AOVY}.

Using \eqref{Zoproduct}, we can compute our vertex using the ordinary topological vertex formalism.
In the next section, we give a yet another way of computing the non-commutative topological vertex. The advantage of our approach is that the final expression manifestly takes a simple infinite product form, and we do not have to worry about the summation of Schur functions.


\section{The Closed Expression for the Vertex}

In this section we give a closed expression for our non-commutative topological vertex. We do this by proving a curious identity stating the equivalence of our vertex for a toric Calabi-Yau manifold $X$ with a closed BPS partition function for an orbifold of $X$\footnote{This is reminiscent of story of \cite{GO1}, where the `bubbling geometry' $X'$ is constructed for given a toric Calabi-Yau manifold $X$ such that the open+closed topological string partition function on $X$ is equivalent to closed topological string partition function on $X'$. However, our story is different in that the vertex computes only a part of the full open BPS partition function; the partition function itself is given by summation of our vertices over representations.}. For another method using vertex operators, see \cite{Nagao3,Sulkowski}.

Start from a non-commutative topological vertex for a generalized conifold $X$, which has $L-1$ compact $\bP^1$'s. As we discussed above, for the definition of the vertex we need 
(1) $\sigma$ for a choice of the crepant resolution of $X$,
(2) a map $\theta$ specifying the chamber together with $\sigma$
(3) a set of representations $\lambda=(\mu,\nu)$,
and the resulting vertex is denoted by $\scC_{\sigma,\theta;\lambda}(q,Q)$. In the following we consider the special case $\nu=(\lambda_+,\lambda_-)=(\emptyset, \emptyset)$.

Choose an $M$ and consider the $\bZ_M$ orbifold $X'$ of $X$. 
We choose the orbifold action such that when the toric diagram of $X$ is a trapezoid with a top and the bottom edge of length $L_+$ and $L_-$ respectively, $X'$ has length $M L_+$ and $M L_-$.
We also choose map
$$
\sigma': \{1/2,3/2,\ldots, ML-1/2\}\to \{\pm 1\}.
$$
and
$$
\theta': \bZ_h\to \bZ_h, \quad \theta'(h+ML)=\theta'(h)+ML
$$
such that
\beq
\sigma\circ \theta=\sigma' \circ\theta', \quad \mu\circ \theta=\emptyset\circ \theta'.
\label{assumptionmain}
\eeq

Then
\beq
\mathcal{C}_{(\sigma,\theta\sss;\sss \emptyset,\mu)}(q,Q)=
\mathcal{C}_{(\sigma',\theta'\sss;\sss \emptyset, \emptyset)}(q',Q')|_{q^{\theta}_i=q'{}^{\theta'}_i=q'{}^{\theta'}_{i+L}=\cdots q'{}^{\theta'}_{i+(M-1)L}},
\label{main2}
\eeq
where $i=0,\ldots, L-1$.
\begin{NB}Modified: See Appendix \ref{app.proof} for the proof of this statement.
The Appendix also discusses \end{NB}%
See Appendix \ref{app.proof} for an explicit method for choosing such $M,\sigma', \theta'$ satisfying \eqref{assumptionmain} as well as generalization of \eqref{main2} to the case of refined BPS invariants.

Since the infinite-product expression for closed BPS partition function for a generalized conifold is already known (section \ref{sec.closed}), we have a closed expression of our vertex when $\nu=\emptyset$.

\section{Examples}
Let us illustrate the above procedure by several examples.

\subsection{$\mathbb{C}^3$}
First, we begin with the non-commutative topological vertex for $\mathbb{C}^3$. Since there is no wall crossing phenomena involved in this case\footnote{There is one wall between $R>0$ and $R<0$, however we do not discuss such a wall since we are specializing to the case $R>0$.}, the vertex should coincide with the ordinary topological vertex, thus providing a useful consistency check of our proposal.
For $\bC^3$, we have $L=1, \sigma(1/2)=-1, \theta=id$.

Take $\mathbb{C}^3$ with representation $\lambda=(\mu,\nu=\emptyset)$ with $\mu=(N,N-1,\ldots, 1)$ at one leg.
The above-mentioned procedure gives
$M=2$, and thus $X=\bC^2/\bZ_2\times \bC$.
The method in Appendix \ref{app.proof} gives
$$
\theta'(1/2)= 1/2-N,\quad  \theta'(3/2)=3/2+N, \quad  \sigma'=-1,
$$
and thus
$$
\left[ B(\alpha_1) \right]=N.
$$
The weight is given by \eqref{conifoldwt1} or \eqref{conifoldwt2}.
By solving for 
$$
q'{}_0^{\theta'}=q'{}_1^{\theta'}=q,
$$
we have 
$$
q'{}_0=q^{-2N+1}, \quad q'{}_1=q^{2N+1}
$$
in the case of $N$ odd.
\begin{NB}Modified
(for $N$ odd; the case of $N$ even is similar)
$$
q'{}_0^{\theta'}=q'{}_0^{-k} q'{}_1^{-k+1}, \quad
q'{}_1^{\theta'}=q'{}_0^{k+1} q'{}_1^k
$$
and by solving for 
$$
q'{}_0^{\theta'}=q'{}_1^{\theta'}=q
$$
we have 
$$
q'{}_0=q^{-2k+1}, \quad q'{}_1=q^{2k+1}
$$
\end{NB}%
Substituting this into the closed BPS partition function
$$
\prod_{n>0}(1-q'{}_0^n q'{}_1^n)^{-2n} 
\prod_{n>0}(1-q'{}_0^n q'{}_1^{n+1})^{-n} \prod_{n>N}(1-q'{}_0^n q'{}_1^{n-1})^{-n},
$$
we have
\begin{NB}corrected\end{NB}%
\beq
M(q)/ \prod_{n=1}^N (1-q^{2n-1})^{(N+1)-n}.
\label{egresult1}
\eeq
This coincides with the know expression for the topological vertex \cite{AKMV}\footnote{In the normalization of \cite{AKMV}, \eqref{egresult1} coincides with 
$M(q)\, q^{-\| \mu^{\tenchi} \| /2 }\, C_{\mu,\emptyset,\emptyset}$}. The case of $N$ even in similar. \begin{NB2}added\end{NB2}

\subsection{Resolved Conifold}
Now let us discuss the next simplest example, the resolved conifold.

\begin{NB}
In this case, $L=2$ and we choose a triangulation given by $\sigma(1/2)=+, \sigma(3/2)=-$. The $N$-th chamber is given by 
$\theta_N(1/2)=1/2-N, \theta(3/2)=2/3+N$.
\end{NB}%

Consider the representation $\mu=({\rm \tiny \yng(2,1)},{\rm \tiny \yng(1)})$, with \begin{NB}changed $\theta=\theta_0=id$\end{NB}%
$\theta=id, \sigma(1/2)=+, \sigma(3/2)=-$. In this case, the method in Appendix \ref{app.proof} gives \begin{NB}Modified\end{NB}%
$$
\theta'(1/2)=-7/2,\quad \theta'(3/2)=-1/2,\quad \theta'(5/2)=11/2,\quad 
\theta'(7/2)=13/2
$$
with
$$
\sigma'(1/2)=+,\quad \sigma'(3/2)=-,\quad \sigma'(5/2)=-,\quad \sigma'(7/2)=+.
$$
Then we have\begin{NB}Modified\end{NB}%
\beq\label{eq_Bforconifold}
[B(\alpha_1)]=[B(\alpha_3)]=0,\quad 
[B(\alpha_2)]=[B(\alpha_2+\alpha_3)]=1,\quad
[B(\alpha_1+\alpha_2)]=[B(\alpha_1+\alpha_2+\alpha_3)]=2,  
\eeq
and 
\beq\label{eq_qforconifold}
q'_0=q_0^{-3}q_1^{-3}=q^{-3}, \quad
q'_1=q_0 q_1^{2}=q Q , \quad
q'_2=q_0^{3}q_1^{3}=q^3, \quad
q'_3=q_0=q Q^{-1}.
\eeq
The closed BPS partition function corresponding to the B-field \eqref{eq_Bforconifold} is the following:
\begin{NB2} changed: $q<-> Q$\end{NB2}
\begin{align}
\prod_{n>0} (1-q'^n)^{-4n} 
\prod_{n>0} (1-q'^n Q'_1)^n \prod_{n>0} (1-q'^n Q'{}_1^{-1})^n
\prod_{n>0} (1-q'^n Q'_2)^{-n} \prod_{n>1} (1-q'^n Q'{}_2^{-1})^{-n}\notag\\
\prod_{n>0} (1-q'^n Q'_3)^n \prod_{n>0} (1-q'^n Q'{}_3^{-1})^n 
\prod_{n>0} (1-q'^n Q'_1 Q'_2)^n \prod_{n>2} (1-q'^n Q'{}_1^{-1} Q'{}_2^{-1})^n \notag\\
\prod_{n>0} (1-q'^n Q'_2 Q'_3)^n \prod_{n>1} (1-q'^n Q'{}_2^{-1} Q'{}_3^{-1})^n
\prod_{n>0} (1-q'^n Q'_1 Q'_2 Q'_3)^{-n} \prod_{n>2} (1-q'^n Q'{}_1^{-1} Q'{}_2^{-1} Q'{}_3^{-1})^{-n}.\label{eq_closedBPS}
\end{align}
Substituting \eqref{eq_qforconifold} for \eqref{eq_closedBPS} under the identification \eqref{chemicalweight}, we obtain the open BPS partition function:\begin{NB}I think the notation $Z^c_{\mathrm{NCDT}}$ is not suitable in physics paper, but I have no idea what is the correct notation in physics context. Please modify it.\end{NB}%
\[
\scZ^c_{\mathrm{NCDT}}(q,Q)\cdot (1-q)^{-3}(1-q^3)^{-1}(1-Q)^2(1-q^2 Q)(1-q^2Q^{-1}),
\]
\begin{NB2}\[
\scZ^c_{\mathrm{NCDT}}(q,Q)\cdot (1-q)^{-3}(1-q^3)^{-1}(1-Q_1)^2(1-q^2Q_1)(1-q^2Q_1^{-1})
\]\end{NB2}
where 
\begin{align*}
\scZ^c_{\mathrm{NCDT}}(q,Q):&=
\prod_{n>0} (1-q^n)^{-2n} 
\prod_{n>0} (1-q^n Q)^n 
\prod_{n>0} (1-q^n Q^{-1})^n \\
&=\scZ^c_{\mathrm{top}}(q,Q) \scZ^c_{\mathrm{top}}(q,Q^{-1}) .
\end{align*}
This coincides with the expression computed from the result of \cite{IK}
$$
M(q)^2 \prod_{n>0} (1-q^n Q) \frac{1}{(1-q)^3 (1-q^3)} (1-Q)^2 (1-Q q^2)(1-Q q^{-2})
$$
up to the wall crossing factor and the normalization by a monomial as remarked in footntote \ref{normalization}.

\section{Discussions}

Let us conclude this paper by pointing out several interesting problems which require further exploration.

\begin{itemize}

\item In this paper, we only discussed wall crossings with respect to the closed string moduli. However, we expect that there should be wall crossing associated with the open string moduli as well. Therefore, the question arises: At which values of the open string moduli is our vertex defined? How does the vertex change as we change the open string moduli? Some of these issues will be discussed in \cite{AY}.

\item It would be interesting to see if there is a generalization of GV large $N$ duality \cite{GVdual} including the background dependence, and whether there is a Chern-Simons interpretation of the wall crossing phenomena (see \cite{CecottiV} for a related idea).

\item It would be interesting to extend our definition and computation to a Calabi-Yau manifold with compact 4-cycles and with multiple D6-branes. \begin{NB2} adde\end{NB2} The former should be possible by combining our vertices, and will be related to the Nekrasov's partition function \cite{Nekrasov}.

\item 
Derivation of open BPS invariants from supergravity viewpoint \cite{DenefM} is another interesting problem. See \cite{Denefhole} for the related discussion in the case of orientifolds. 
Another related question is the connection of the crystal melting expansion of $\scZ^o_{\rm BPS}$ with open version of OSV conjecture \cite{ANV}. 

\end{itemize}


\section*{Acknowledgments}
\begin{NB2} changed \end{NB2}
M.~Y. would like to thank M.~Aganagic and D.~Krefl for collaboration on related projects.
He would also like to thank M.~Cheng, J.~Gomis, L~Hollands, D.~Krefl,
T.~Okuda, H.~Ooguri, N.~Saulina, P.~Sulkowski, C.~Vafa and J.~Walcher for stimulating discussions. 

K.~N. is supported by JSPS Fellowships for Young Scientists (No.\ 19-2672).
M.~Y. is supported by DOE grant DE-FG03-92-ER40701, by the JSPS fellowships for Young Scientists, by the 
World Premier International Research Center Initiative, 
and by the 
Global COE Program for Physical Sciences Frontier at 
the University of Tokyo, both by MEXT of Japan. 
M.Y. would also like to thank Centro de Ciencias de Benasque Pero Pascul, Simons Center for Geometry and Physics at Stony Brook and Berkeley Center for Theoretical Physics for hospitality, where part of this work has been performed.


\appendix

\section{The Parametrization of Chambers by the Weyl Group} \label{app.Weyl}
In this section we explain the parameterization of chambers of closed BPS invariants by maps $\sigma$ and $\theta$, as claimed in the main text.

The map $\theta$ is defined to be a map from the set of half-integers $\bZ_h$ to itself
$$
\theta: \bZ_h\to \bZ_h,
$$
satisfying the following two conditions.
First,
\beq
\theta(h+L)=\theta(h)+L \label{theta+L}
\eeq
for any $h\in\bZ_h$. In other words, $\theta$ is periodic with period $L$.
Second,
\beq
\sum_{i=1}^L \theta\left(i-\frac{1}{2}\right)=\sum_{i=1}^L \left(i-\frac{1}{2}\right).
\label{thetasum}
\eeq
Therefore $\theta$ is specified by \begin{NB}corrected\end{NB}%
$L-1$ (half-)integers, namely \begin{NB}corrected\end{NB}%
$L$ half-integers $\theta(1/2),\ldots, \theta(L-1/2)$, subject to one constraint \eqref{thetasum}. 
Let us assume for the moment that $\theta$ satisfies the condition \begin{NB}I change the notation $\theta\to \theta^{-1}$\end{NB}%
\beq
\theta\left(\frac{1}{2}\right)< \theta\left(\frac{3}{2}\right) < \ldots 
< \theta\left(L-\frac{1}{2}\right).
\label{thetaincrease}
\eeq
We will discuss other cases later.

Given $\sigma$, we have a specific choice of resolution having \begin{NB}corrected\end{NB}%
$L-1$ $\bP^1$'s.
Moreover, given a map $\theta$ we can determine the corresponding value of the B-field $B_{\theta}$ by
\begin{NB}corrected
\beq
\left[ B_{\theta}(\alpha_i+\ldots +\alpha_j) \right]=
\left[ \frac{\theta^{-1}(j-1/2)-\theta^{-1}(i-1/2)}{L} \right].  \label{thetaB}
\eeq
\end{NB}%
\beq
\left[ B_{\theta}(\alpha_i+\ldots +\alpha_j) \right]=
\sharp \{m\in\Z\mid \theta(i-1/2)<mL<\theta(j+1/2)\}.  \label{thetaB}
\eeq
It is easy to see that this gives a well-defined values of the integer parts of the B-field, which parametrize the chamber as explained in the main text.
Conversely, it can also be proven that given any B-field, we can find a corresponding $\theta$ uniquely\footnote{This comes from the fact that the action of the affine Weyl group on the space of B-fields is faithful.}.

When the condition \eqref{thetaincrease} is not satisfied, we need to change the choice of the crepant resolution.
Choose a permutation $\Sigma$ of $\{ 1/2,\ldots, L-1/2\}$ such that
\begin{NB}I change the notation $\theta\to \theta^{-1}$\end{NB}%
\beq
\theta\left(\Sigma\left(\frac{1}{2}\right)\right)< \theta\left(\Sigma\left(\frac{3}{2}\right)\right) < \ldots 
< \theta\left(\Sigma\left(L-\frac{1}{2}\right)\right)
\label{permdef}
\eeq
holds. Then we replace $\theta$ by \begin{NB}corrected\end{NB}%
$\theta':=\theta\circ \Sigma$ \begin{NB}Modified\end{NB}%
and we choose $\sigma'$ so that $\sigma\circ\theta=\sigma'\circ\theta'$\footnote{In fact, we can take $\sigma\circ\underline{\theta}\circ\Sigma^{-1}\circ\underline{\theta}^{-1}$ as $\sigma$, where $\underline{\theta}$ is the permutation induced by $\theta$.}.
Note that the combination $\sigma\circ\theta$, which appears in the definition of the vertex in section \ref{sec.def}, remain invariant under this process.
This means that sometimes different $\sigma,\theta$ and $\sigma', \theta'$ corresponds to the same chamber. We can either fix $\sigma$ and change $\theta$ to parametrize chambers, or change both $\sigma$ and $\theta$ for convenience. The latter parametrization is redundant, but sometimes useful.

\bigskip

In the above discussion, $\theta$ appears somewhat artificially, but $\theta$ is often used in the mathematical literature.
The reason is that the maps $\theta$ makes a group, which is the Weyl group of $\hat{A}_{L-1}$.
As is well-known, the Weyl group of $A_{L-1}$ is the $L$-th symmetric group, which is a set of isomorphism $\{1/2,\ldots, L-1/2 \}\to \{1/2,\ldots, L-1/2 \}$.
The map $\theta$ gives a generalization to the affine case.
The affine Lie algebra $\hat{A}_{L-1}$ appears in the formula for the BPS partition function \cite{Nagao1,AOVY}, and as we have seen specifies a chamber structure. This is reminiscent of the appearance of the Weyl group of the Borcherds-Kac-Moody algebra in $\mathcal{N}=4$ wall crossing \cite{ChengV,CD}. It would be interesting to explore this point further.

\bigskip
Finally, let us illustrate this parametrization with examples. Consider the resolved conifold ($L=2$) with the resolution $\sigma(1/2)=+, \sigma(3/2)=-$.
Due to the condition \eqref{thetasum}, we can write\begin{NB}I change the notation\end{NB}%
\beq
\theta(1/2)=1/2-N,\quad \theta(3/2)=3/2+N.
\label{theta1}
\eeq
and this integer $N$ parametrize the chambers. This integer $N$ is the same integer $N$ appearing in section \eqref{thetaN}.
When $N\ge 0$, condition \eqref{thetaincrease} is satisfied and we are in one resolution $\sigma$. When $N<0$, \eqref{thetaincrease} is not satisfied and by a flop transition we are in a different crepant resolution specified by $\sigma'(1/2)=-, \sigma'(3/2)=+$.
If we define $\theta'$ by\begin{NB}I change the notation\end{NB}%
$$
\theta'(1/2)=1/2+N, \quad \theta'(3/2)=3/2-N,
$$
then $\sigma,\theta$ for $N>0$ and $\sigma',\theta'$ for $N<0$ parametrize the same chamber. In this sense, we can either choose $\sigma$ to be fixed and change $\theta$, or change both $\sigma$ and $\theta$, although the latter is a redundant parametrization.
The chamber corresponding to topological string theory for one resolution $\sigma$ is given by $N\to\infty$ in \eqref{theta1}, and another resolution $\sigma'$ given by $N\to -\infty$.

As a next example suppose $L=3$.
When \begin{NB}Modified\end{NB}%
$$
\theta(1/2)=-5/2, \quad \theta(3/2)=3/2, \quad \theta(5/2)=11/2,
\label{theta3-1}
$$
we have
$$
\left[B_{\theta}(\alpha_1)\right]=\left[B_{\theta}(\alpha_2)\right]=1, \quad \left[B_{\theta}(\alpha_1+\alpha_2)\right]=2.
$$
This is also given by\begin{NB}Modified\end{NB}%
$$
\theta(1/2)=11/2, \quad \theta(3/2)=3/2, \quad \theta(5/2)=-5/2.
\label{theta3-2}
$$
However, they parametrize different chambers in general\footnote{Sometimes they give the same chamber. This happens, for example for $\bC^2/\bZ_2\times \bC$, where there is a unique choice of crepant resolution.}. This is because the two $\theta$'s are related by a permutation $\Sigma: \{1/2.3/2,5/2\}\mapsto \{5/2,3/2,1/2 \}$, and correspondingly we have to change the choice of crepant resolution $\sigma$ by $\sigma\circ \Sigma$ as mentioned around \eqref{permdef}.

More generally, if we have \begin{NB}Modified\end{NB}%
$$
\theta\left(\frac{1}{2}\right)= \frac{1}{2}-2 L_1- L_2, \quad
\theta\left(\frac{3}{2}\right)= \frac{3}{2}+ L_1-L_2, \quad
\theta\left(\frac{5}{2}\right)= \frac{5}{2} +L_1 +2L_2.
$$
then
$$
\left[B_{\theta}(\alpha_1)\right]=L_1, \quad \left[B_{\theta}(\alpha_2)\right]=L_2, \quad \left[B_{\theta}(\alpha_1+\alpha_2)\right]=L_1+L_2,
$$
\begin{NB}Added\end{NB}%
and if we have
$$
\theta\left(\frac{1}{2}\right)= -\frac{1}{2}-2 L_1- L_2, \quad
\theta\left(\frac{3}{2}\right)= \frac{3}{2}+ L_1-L_2, \quad
\theta\left(\frac{5}{2}\right)= \frac{7
}{2} +L_1 +2L_2.
$$
then
$$
\left[B_{\theta}(\alpha_1)\right]=L_1, \quad \left[B_{\theta}(\alpha_2)\right]=L_2, \quad \left[B_{\theta}(\alpha_1+\alpha_2)\right]=L_1+L_2+1.
$$


\section{Young Diagrams and Maya Diagrams} \label{app.core}

A Young diagram $\lambda$ is a set of non-increasing positive integers 
$\lambda=(\lambda_1,\lambda_2,\ldots )$. As is well-known, this is equivalently 
represented by a Maya diagram, namely a map
$$
\lambda: \bZ_h\to \{\pm 1\}
$$
such that $\lambda(h)=\pm 1$ for $\pm h \gg  1$.
We sometimes represent $\lambda$ by
$$
\lambda=\cdots \quad \lambda(-5/2) \quad \lambda(-3/2) \quad \lambda(-1/2) \quad \Big| \quad \lambda(1/2) \quad \lambda(3/2) \quad \lambda(5/2)\quad \cdots,
$$
where the symbol $\Big|$ represents the position of the origin.
For notational simplicity, we use the same symbol $\lambda$ for a Maya diagram.
Our convention is shown in Figure \ref{Mayaconv}.

\begin{figure}[htbp]
\centering{\includegraphics[scale=0.45]{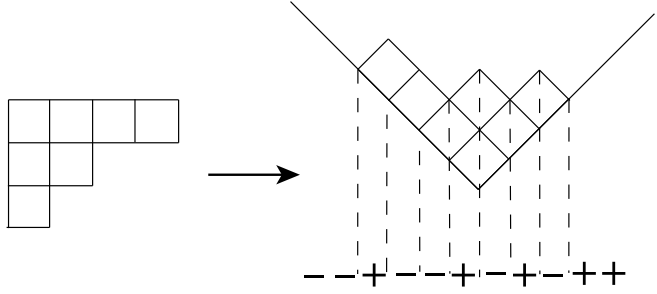}}
\caption{The convention of the Maya diagram in this paper.}
\label{Mayaconv}
\end{figure}

For example, 
\begin{eqnarray*}
{\rm \tiny \yng(2,1,1)}&=\cdots----+-\Big|++-+++ \cdots,\\
{\rm \tiny \yng(3,1)}&=\cdots ---+--\Big|+-++++ \cdots. \label{yng31}
\end{eqnarray*}

For a Young diagram and a positive integer $M$, define the quotients 
$\lambda_i (i=1,\ldots,M)$ by
\beq
\lambda_i (h)=\lambda\left(i-1/2+(h-1/2)M\right) \quad \textrm{for}\quad h\in \bZ_h.
\eeq

As an example, let us consider $\lambda={\tiny \yng(3,1)}$ shown in \eqref{yng31}. If you take $M=2$, then $M=2$-quotients are given by
\begin{align*}
\lambda_1&=\cdots------\Big| ++++++\cdots, \\
\lambda_2&=\cdots----+-\Big| -+++++\cdots.
\end{align*}

Suppose that $M$ is chosen such that the representation $\lambda_i$ is trivial for all $i=1,\ldots, M$.
This means that $\lambda_i$ 
can be written as 
\beq
\lambda_i (h)=\emptyset( h+ M N(i)).
\eeq
The integers $N(j)$ are called $M$-cores of $\lambda$, and satisfies
$$
\sum_{i=1}^M N(i)=0.
$$
For example, if we take $M=3$ for \eqref{yng31}, 
\begin{align*}
\lambda_1&=\cdots-----+\Big| ++++++\cdots, \\
\lambda_2&=\cdots------\Big| -+++++\cdots, \\
\lambda_3&=\cdots------\Big| ++++++\cdots,
\end{align*}
and 
\begin{align*}
N(1)=-1, \quad N(2)=1, \quad N(3)=0.
\end{align*}

\section{Proof of \eqref{main2}}\label{app.proof}

In this Appendix we give a proof of \eqref{main2}.
First, the following is clear from the definition:
\begin{prop}\label{prop.moregeneral}
Let $\sigma, \theta$ to be maps specifying the chamber for a Calabi-Yau manifold $X$.
Choose an integer $M$ and $\sigma', \theta'$ for $X'=X/\bZ_M$
such that the following condition holds:
\begin{equation*}
\sigma\circ\theta=\sigma'\circ\theta',\quad \mu\circ\theta=\mu'\circ\theta'.
\end{equation*}
Then we have
\[
\mathcal{C}_{(\sigma',\theta'\sss;\sss\mu',\nu')}(q',Q')|_{q^{\theta}_i=q'{}^{\theta'}_i=q'{}^{\theta'}_{i+L}=\cdots q'{}^{\theta'}_{i+(M-1)L}}
=\mathcal{C}_{\type}(q,Q).
\label{propassumption}
\]
\end{prop}
This is because the both sides of the equation are defined by the same crystal with the same weights. The crystal for an orbifold is the same as the original crystal, the only difference being the difference of the weights; the crystal for the orbifold has more colors (variables). 
 However, in the above equation we are specializing the variables so the weights are also the same.
\begin{lem}\label{lem.choose}
For any $\sigma, \theta, \lambda$, we can take $M, \sigma', \theta'$ such that
\[
\sigma\circ\theta=\sigma'\circ\theta',\quad \mu\circ\theta=\emptyset\circ\theta'.
\]
\end{lem}
\begin{proof}
We choose an integer $M$ such that all the $M$-quotients\footnote{See Appendix \ref{app.core} for the definition of $M$-quotients.} of $\mu_i$'s become trivial, i.e. $ML$-quotients of the combined representation $\mu$ (see \eqref{combinedmu}) is trivial.
For example, this is satisfied if we define
\[
h_-:=\mathrm{min}\{h\in\Zh \mid \mu\circ \theta(h)=+\},\quad 
h_+:=\mathrm{max}\{h\in\Zh \mid \mu\circ \theta(h)=-\},
\]
and take $M$ so that $ML>h_+-h_-$\footnote{The choice of $M$ is not unique. 
But the final result is independent of the choice of $M$. For practical computation it is useful to take the minimum $M$.}.
This means that for any half-integer $1/2\leq h\leq ML-1/2$ we can take $N(h)\in \Z$ such that 
\begin{align*}
\mu\circ \theta(h+NML)&=
\begin{cases}
- & (N<N(h)),\\
+ & (N\geq N(h))
\end{cases} \\
&=\emptyset \left(h+(N-N(h))ML\right).
\end{align*}
In other words, $N(j)$ is the $ML$-core\footnote{See Appendix \ref{app.core} for the definition of $M$-quotients.} of $\mu$.

Therefore the second condition of \eqref{propassumption} holds 
if we define $\theta'\colon\Zh\to\Zh$ by 
\[
\theta'(h)=h-N(h)ML\]
\begin{NB}removed: 
or
\[
\theta'^{-1}(h)=h+N(h)ML
\]
\end{NB}%
for $h\in \bZ_h$\footnote{For practical computations, it is useful to further perform a permutation to $\theta$ such that \eqref{thetaincrease} holds. See Appendix \ref{app.Weyl}.}.
It is clear that the first condition of \eqref{propassumption} determines $\sigma'$ uniquely.
\end{proof}

Our theorem follows from Proposition \ref{prop.moregeneral} and Lemma \ref{lem.choose}.
\begin{thm}
For $\sigma, \theta, \lambda$, take $M, \sigma', \theta'$ as above. 
Then we have
\[
\mathcal{C}_{(\sigma,\theta\sss;\sss \emptyset,\mu)}(q,Q)=
\mathcal{C}_{(\sigma',\theta'\sss;\sss \emptyset, \emptyset)}(q',Q')|_{q^{\theta}_i=q'{}^{\theta'}_i=q'{}^{\theta'}_{i+L}=\cdots q'{}^{\theta'}_{i+(M-1)L}}.
\]
\end{thm}
\begin{NB2}added\end{NB2}
It is straightforward to generalize this theorem to the case of refined BPS invariants discussed in section \ref{subsec.refined}.

\bibliographystyle{JHEP}
\bibliography{dthesis}

\end{document}